\documentclass[11pt]{scrartcl}
\usepackage{amsmath,amssymb,amsthm}
\usepackage{hyperref}
\usepackage{fullpage}
\usepackage{amsmath,amsthm}
\usepackage{newtxtext}
\usepackage{newtxmath}

\newtheorem{claim}{Claim}

\newtheorem{theorem}{Theorem}

\newcommand {\calC} {{\cal C}}

\newcommand {\set} [1] {\left\{ #1 \right\}}

\newcommand {\Exp} {\mathbb{E}}

\newcommand {\supp} {\operatorname{supp}}
\DeclareMathOperator {\argmax}{arg\,max}

\title{Sharp Analysis of Gaussian Rounding for Boolean Max $k$-CSP}
\author{Yury Makarychev (TTIC)}
\date{}
\begin{document}
\maketitle

\begin{abstract}
In this note, we show that the approximation algorithm for Boolean Max $k$-CSP presented in~\cite{MM} yields a $(1-o_k(1))k/2^k$ approximation, as conjectured in~\cite{MMSurvey}. This improves the previous guarantee of $(0.626612-o_k(1))k/2^k$ from~\cite{MM} and asymptotically matches the known hardness results; see~\cite{MMSurvey,am09,DeMossel,chan12}. The result is a short corollary of the Gaussian stochastic domination theorem of Mulgund~\cite{Mulgund}. 
\end{abstract}

\section{Introduction}

For Boolean Max $k$-CSP, Hast~\cite{hast05} gave an $\Omega(k/(2^k\log k))$ approximation, and Charikar, Makarychev, and Makarychev~\cite{cmm} obtained a $c k/2^k$ approximation for an absolute constant $c>0.44$. In~\cite{MM} we improved the constant to $0.626612-o_k(1)$ and conjectured that the same approach gives a $(1-o_k(1))k/2^k$ approximation; see also the survey~\cite{MMSurvey}.

The approximation guarantee is asymptotically optimal. Austrin and Mossel~\cite{am09} showed, assuming the Unique Games Conjecture, hardness of approximation within $(k+O(k^{0.525}))/2^k+\varepsilon$. De and Mossel~\cite{DeMossel} subsequently improved this to $(k+1)/2^k+\varepsilon$ for odd $k$ and $(k+2)/2^k+\varepsilon$ for even $k$. Thus, under the Unique Games Conjecture, the optimal approximation factor is $(1+o_k(1))k/2^k$. Chan~\cite{chan12} obtained unconditional hardness from pairwise-independent subgroups; in particular, for $k=2^r-1$, his result gives hardness within $(k+1)/2^k+\varepsilon$. Hence the constant in our approximation guarantee is optimal under $P\ne NP$ for an infinite sequence of values of $k$.

\noindent\textbf{Overview.}
We use the same SDP relaxation, notation, and Gaussian rounding as in~\cite{MM}. The only new ingredient is Mulgund's stochastic domination theorem~\cite{Mulgund}, which we use as a black box. If $C$ is a clause and $\|z_C\|^2=\beta^2/k$, the theorem gives the key estimate
\[
\Pr\left(C\text{ is satisfied}\right)\ge \int_{\mathbb R}F(cx)^k\,d\gamma(x),
\qquad
c^2=\frac{\beta^2-1}{k-\beta^2}.
\]
At $\|z_C\|^2=1/k$ the right-hand side is exactly $2^{-k}$, while slightly above this value it grows sufficiently fast. Combining this rounding with a random assignment gives the claimed $(1-o_k(1))k/2^k$ approximation. Thus the result is essentially a corollary of~\cite{Mulgund}.

\noindent{\textbf{Related work:}} Concurrently and independently, Bakshi~\cite{Bakshi} also settled the conjecture for Boolean Max $k$-CSP.

\section{The algorithm}

An instance of Boolean Max $k$-CSP consists of Boolean variables $x_u\in\{1,2\}$ and constraints, each depending on at most $k$ variables; the goal is to find an assignment that satisfies the maximum number of constraints. Following Trevisan~\cite{trev98}, as in~\cite{MM}, we may replace a predicate $P(x_{v_1},\ldots,x_{v_r})$ by the clauses
\[
\set{(x_{v_1}=i_1)\wedge\cdots\wedge(x_{v_r}=i_r): P(i_1,\ldots,i_r)=1}.
\]
For every assignment, exactly one of these clauses is satisfied if $P$ is satisfied, and none is satisfied otherwise. Thus this transformation preserves the value of every assignment. We thus assume below that every constraint is a conjunction of literals.
We use the same SDP relaxation and notation as in~\cite{MM}. For every variable $x_u$ and $i\in\{1,2\}$ the SDP has a vector $u_i$, and for every clause $C$ it has a vector $z_C$:
\begin{align*}
\text{maximize: }&\sum_{C\in\calC}\|z_C\|^2\\
\text{subject to }&\\
&\|u_1\|^2+\|u_2\|^2=1 \\
&\langle u_1,u_2\rangle=0,\\
&\langle u_i,z_C\rangle=\|z_C\|^2 &&\text{if }(u,i)\in C,\\
&\langle u_j,z_C\rangle=0 &&\text{if }(u,i)\in C,\ j\ne i.
\end{align*}
The intended meaning of these vectors is the following. In an integral solution, for a fixed unit vector $\mathbf e$, we set $u_i=\mathbf e$ if $x_u=i$ and $u_i=0$ otherwise; similarly, $z_C=\mathbf e$ if $C$ is satisfied and $z_C=0$ otherwise. Thus $\|u_i\|^2$ is a relaxation of the indicator that $x_u=i$, while $\|z_C\|^2$ is a relaxation of the indicator that $C$ is satisfied. The constraints above are satisfied by every integral assignment, therefore, $SDP\ge OPT$. We may also assume that every clause has length exactly $k$, by padding shorter clauses with extra variables.

\paragraph{Rounding procedure} The SDP rounding is the same Gaussian rounding as in~\cite{MM}: choose a standard Gaussian vector $g$ and set
\[
x_u=\argmax_{i\in\{1,2\}}\langle u_i,g\rangle.
\]
The final algorithm uses this rounding with probability $1/k$ and a uniformly random assignment with probability $1-1/k$. Write
\[
F(t)=\Pr\left(g\le t\right),
\]
where $g\sim {\cal N}(0,1)$, and denote standard Gaussian measure on $\mathbb R$ by $\gamma$.
We use the following theorem of Mulgund~\cite{Mulgund}.

\begin{theorem}[Mulgund]\label{thm:mulgund}
Let $R$ be a $k\times k$ correlation matrix and let $J$ be the all-ones matrix. If
\[
R-\frac1kJ\succeq0
\]
and $X\sim {\cal N}(0,R)$, then, for every $t\in\mathbb R$,
\[
\Pr\left(X_u\le t\text{ for every }u\right)\ge F(t)^k.
\]
\end{theorem}

\section{Analysis of one clause}

Fix a clause $C$. As in~\cite{MM}, we may assume without loss of generality that
$C=\set{(u,1):u\in\supp(C)}$, that is, the clause is of the form $\bigwedge_{u\in\supp(C)}(x_u=1)$.
Crucially, the clause $C$ is satisfied if and only if
\[
\langle v_u,g\rangle\ge0\qquad\text{for every }u\in\supp(C)
\]
where $g \sim {\cal N}(0, I)$ is the random Gaussian vector from the rounding scheme. Our goal now is to lower bound the probability of this event.

Define $s=\|z_C\|^2$
and, for $u\in\supp(C)$,
\[
v_u=u_1-u_2 \qquad\text{and}\qquad w_u=v_u-z_C.
\]
By the SDP constraints, we have $\|v_u\|^2=\|u_1\|^2 + \|u_2\|^2=1$, $\langle v_u,z_C\rangle=s$, $w_u\perp z_C$, and $\|w_u\|^2=1-s$. Let $\Sigma$ be the Gram matrix of the vectors $v_u$, $u\in\supp(C)$. 
Since all vectors $v_u$ share the same component $z_C$, namely $v_u = z_C + w_u$ for $u\in\supp(C)$, where $w_u \perp z_C$, we can write
$$\langle v_a, v_b\rangle = s + \langle w_a, w_b\rangle,$$
or, for $s < 1$, in the matrix form 
\begin{equation}\label{eq:sigma}
\Sigma=sJ+(1-s)G,
\end{equation}
where $G$ is defined as
\[
G_{uv}=\frac{\langle w_u,w_v\rangle}{1-s}.
\]

Note that the correlation matrix of random variables $\xi_u = \langle v_u,g\rangle$ is exactly $\Sigma$. We would like to apply Theorem~\ref{thm:mulgund} now. However, applying it directly to $\xi_u$ will not yield the desired bound. Let $\beta = \sqrt{sk}$ and assume that $\beta \ge1$. We add a part of the common component $sJ$ to $G$, so that after normalization the resulting matrix dominates $J/k$, and keep the rest separate:
\[
\Sigma=(1-s)\left(G+\frac1{k-1}J\right)+\frac{ks-1}{k-1}J.
\]
Define
\[
R=\frac{k-1}{k}\left(G+\frac1{k-1}J\right)
=\frac{k-1}{k}G+\frac1kJ.
\]
Then $R$ is a correlation matrix and
\[
R-\frac1kJ=\frac{k-1}{k}G\succeq0.
\]
Let $\theta=\frac{k-\beta^2}{k-1}$, using $s=\beta^2/k$, we get
\[
\Sigma=\theta R+(1-\theta)J.
\]
Consequently, if $X\sim {\cal N}(0,R)$ and $Z\sim {\cal N}(0,1)$ are independent, then
\[
\bigl(\xi_u\bigr)_{u\in\supp(C)}
\stackrel{d}{=}
\bigl(\sqrt\theta X_u+\sqrt{1-\theta}Z\bigr)_{u\in\supp(C)}.
\]
Applying Theorem~\ref{thm:mulgund} after conditioning on $Z$ gives the key estimate
\begin{align}
&\Pr\left(C\text{ is satisfied}\right) = \Pr\left(-X_u \leq \sqrt{\frac{1-\theta}{\theta} Z}\text{ for all }u\in\supp(C)\right)
\ge \int_{\mathbb R}F(cx)^k\,d\gamma(x),\label{eq:key-integral}\\
&\text{where } c^2=\frac{1-\theta}{\theta}
=\frac{\beta^2-1}{k-\beta^2}.
\label{eq:c}
\end{align}

\begin{claim}\label{claim:boolean-success}
Let
\[
\beta_k^2=1+\frac{4\pi\log k}{k}.
\]
For all sufficiently large $k$, if
\[
\|z_C\|^2\ge\frac{\beta_k^2}{k},
\]
then the SDP rounding satisfies $C$ with probability at least $k^2/2^k$.
\end{claim}

\begin{proof}
It suffices to estimate~\eqref{eq:key-integral} at $\beta=\beta_k$. Indeed, $c$ is increasing in $\beta$, and the integral is increasing in $c\ge0$: pairing $x$ and $-x$, the integrand becomes $p^k+(1-p)^k$, where $p=F(cx)\ge1/2$, and this is increasing in $p$. Let $c_k$ be the value of $c$ corresponding to $\beta=\beta_k$. Then
\[
c_k^2=(4\pi+o(1))\frac{\log k}{k^2}.
\]
Put $a=\sqrt{2/\pi}$ and $x_k=ak c_k$. Since
$
\log(2F(t))=at+O(t^2)
$ (when $t\to 0$)
and $c_kx_k=ak c_k^2=o(1)$, we have
\begin{align*}
\log\left(2^k\int_{\mathbb R}F(c_kx)^k\,d\gamma(x)\right)
&\ge k\log(2F(c_kx_k))+\log\gamma([x_k,x_k+1])
\ge a^2k^2c_k^2-\frac{x_k^2}{2}+o(k^2c_k^2)\\
&=\left(\frac1\pi+o(1)\right)k^2c_k^2 
=(4+o(1))\log k.
\end{align*}
Hence
\[
\int_{\mathbb R}F(c_kx)^k\,d\gamma(x)
\ge \frac{k^{4-o(1)}}{2^k}\ge\frac{k^2}{2^k}.
\]
The case $\|z_C\|=1$ is immediate.
\end{proof}

\section{Approximation guarantee}

\begin{theorem}\label{thm:main}
There is a randomized polynomial-time approximation algorithm for Boolean Max $k$-CSP with approximation guarantee
\[
\left(1-O\left(\frac{\log k}{k}\right)\right)\frac{k}{2^k}.
\]
In particular, the approximation guarantee is $(1-o_k(1))k/2^k$.
\end{theorem}

\begin{proof}
Let
\[
\alpha_k=\frac{1-1/k}{\beta_k^2}=1-O\left(\frac{\log k}{k}\right).
\]
Consider a clause $C$ and let $s=\|z_C\|^2$.
If $s<\beta_k^2/k$, the random assignment satisfies $C$ with probability at least
\[
\frac{1-1/k}{2^k}
\ge \alpha_k\frac{ks}{2^k}.
\]
If $s\ge\beta_k^2/k$, Claim~\ref{claim:boolean-success} and the rounded assignment satisfies $C$ with probability at least
\[
\frac1k\cdot\frac{k^2}{2^k}
=\frac{k}{2^k}
\ge\frac{ks}{2^k}
\ge\alpha_k\frac{ks}{2^k}.
\]
Thus every clause $C$ is satisfied with probability at least
\[
\alpha_k\frac{k}{2^k}\|z_C\|^2.
\]
Summing over all clauses,
\[
\Exp[\text{number of satisfied clauses}]
\ge \alpha_k\frac{k}{2^k}\sum_{C\in\calC}\|z_C\|^2
=\alpha_k\frac{k}{2^k}SDP
\ge\alpha_k\frac{k}{2^k}OPT.
\]
\end{proof}

\bibliographystyle{plain}
\bibliography{mybib2026}

\end{document}